\documentclass[11pt]{article}

\usepackage{fullpage}
\usepackage[linktocpage]{hyperref}
\usepackage[dvipsnames]{xcolor}
\definecolor{ForestGreen}{rgb}{0.0333,0.4451,0.0333}
\definecolor{DarkRed}{rgb}{0.65,0,0}
\definecolor{Red}{rgb}{1,0,0}
\hypersetup{
    unicode=false,          
    colorlinks=true,        
    linkcolor=BrickRed,          
    citecolor=OliveGreen,        
    filecolor=magenta,      
    urlcolor=cyan           
}

\newcommand{\eps}{\epsilon}

\newcommand{\poly}{\operatorname{poly}} 

\usepackage{amsmath}
\usepackage{amssymb}
\usepackage{amsthm}
\usepackage{thmtools}
\usepackage{thm-restate}
\usepackage{bm}
\usepackage{graphicx}
\usepackage{caption}
\usepackage[margin = 1in]{geometry}
\usepackage{natbib}
\setcitestyle{numbers,square}
\usepackage{color}
\usepackage{media9}
\usepackage{subcaption}
\usepackage{comment}
\usepackage{enumitem}
\usepackage[]{hyperref}
\usepackage[nameinlink,capitalize,nosort]{cleveref}
\crefname{claim}{Claim}{Claims}
\usepackage{palatino}

\usepackage{algorithm}
\usepackage[noend]{algpseudocode}

\usepackage{tcolorbox}
\usepackage[colorinlistoftodos]{todonotes}

\definecolor{mygreen}{rgb}{0,0.65,0.66}

\newcommand{\U}{U}
\newcommand{\sep}{\operatorname{sep}}
\newcommand{\spa}{\operatorname{spars}}
\global\long\def\disperse{\mathrm{disperse}}%

\newcommand{\mcD}{\mathcal{D}}

\newcommand{\copies}{\text{copies}}

\renewcommand{\l}{\ell}

\graphicspath{{imgs/}}

\makeatletter
\def\mathcolor#1#{\@mathcolor{#1}}
\def\@mathcolor#1#2#3{%
  \protect\leavevmode
  \begingroup
    \color#1{#2}#3%
  \endgroup
}
\makeatother

\newcommand*{\vsepfbox}[1]{%
  \begingroup
    \sbox0{\fbox{#1}}%
    \setlength{\fboxrule}{0pt}%
    \mbox{\kern-\fboxsep\fbox{\unhbox0}\kern-\fboxsep}%
  \endgroup
}

\theoremstyle{plain} \numberwithin{equation}{section}
\newtheorem{theorem}{Theorem}[section]
\numberwithin{theorem}{section}

\newtheorem{lemma}[theorem]{Lemma}

\newtheorem{definition}[theorem]{Definition}

\theoremstyle{definition}

\usepackage{thmtools} 
\usepackage{thm-restate}

\makeatletter
\newcommand{\subalign}[1]{%
  \vcenter{%
    \Let@ \restore@math@cr \default@tag
    \baselineskip\fontdimen10 \scriptfont\tw@
    \advance\baselineskip\fontdimen12 \scriptfont\tw@
    \lineskip\thr@@\fontdimen8 \scriptfont\thr@@
    \lineskiplimit\lineskip
    \ialign{\hfil$\m@th\scriptstyle##$&$\m@th\scriptstyle{}##$\hfil\crcr
      #1\crcr
    }%
  }%
}
\makeatother

\newcounter{note}[section]

\newcommand{\authorFontSize}{\normalsize}

\title{Simple Length-Constrained Expander Decompositions}
\author{
\authorFontSize
\begin{tabular}{@{\hskip 1.5em}c@{\hskip 3em}c@{\hskip 1.5em}}
  Greg Bodwin\thanks{Supported in part by NSF grant AF-2153680.} &
  Bernhard Haeupler\thanks{This research was partially funded by the Ministry of Education and Science of Bulgaria (support for INSAIT, part of the Bulgarian National Roadmap for Research Infrastructure) and by the European Research Council (ERC) under the European Union's Horizon~2020 research and innovation program (ERC grant agreement 949272).} \\[0.1em]
  University of Michigan &
  INSAIT, University of Sofia ``St.\ Kliment Ohridski'' \\[-0.2em]
   & and ETH Zürich \\[0.8em]
  D Ellis Hershkowitz\thanks{Supported in part by NSF grant CCF-2403236.} &
  Zihan Tan \\[0.1em]
  Brown University &
  University of Minnesota
\end{tabular}
}

\date{}
\begin{document}
\maketitle

\begin{abstract}
 Length-constrained expander decompositions are a new graph decomposition that has led to several recent breakthroughs in fast graph algorithms. Roughly, an $(h, s)$-length $\phi$-expander decomposition is a small collection of length increases to a graph so that nodes within distance $h$ can route flow over paths of length $hs$ while using each edge to an extent at most $1/\phi$. Prior work showed that every $n$-node and $m$-edge graph admits an $(h, s)$-length $\phi$-expander decomposition of size $\log n \cdot s n^{O(1/s)} \cdot \phi m$.

 In this work, we give a simple proof of the existence of $(h, s)$-length $\phi$-expander decompositions with an improved size of $s n^{O(1/s)}\cdot \phi m$. Our proof is a straightforward application of the fact that the union of sparse length-constrained cuts is itself a sparse length-constrained cut. In deriving our result, we improve the loss in sparsity when taking the union of sparse length-constrained cuts from $\log ^3 n\cdot s^3 n^{O(1/s)}$ to $s\cdot n^{O(1/s)}$. 
\end{abstract}

\thispagestyle{empty}

\newpage

\setcounter{page}{1}

\section{Introduction}
Length-constrained expander decompositions are an emerging paradigm which has led to several recent breakthroughs in fast graph algorithms \cite{haeupler2022hop,haeupler2024new,haeupler2023parallel,haeupler2024dynamic,haeupler2025length,haeupler2024low, haeupler2025parallel,haeupler2025cut}. For example, recent work gave algorithms that use length-constrained expander decompositions on $n$-node graphs with $m$ edges to:
\begin{itemize}
    \item Solve $O(1)$-approximate min cost $k$-commodity flow in time $(m+k)^{1+\eps}$ for any constant $\eps > 0$ \cite{haeupler2024low} ; prior to this work, there was no $o(mk)$ time  $o(\log n)$-approximation.
    \item  Give $O(1)$-approximate deterministic distance oracles with $O(\log \log n)$ query time and $n^\eps$-worst case update time for any constant $\eps > 0$ \cite{haeupler2024dynamic}; prior to this work, there were no such $o(n)$-approximate, $n^{o(1)}$ query time and worst case subquadratic update time constructions.
    \item Solve undirected $(1+\eps)$-approximate min cost flow in parallel with $\hat{O}(m)$ work and  $\hat{O}(1)$ depth for any constant $\eps > 0$ \cite{haeupler2025parallel}; all prior $\hat{O}(m)$ work algorithms  had $\Omega(m)$ depth.\footnote{We use $\hat{O}$ to hide $n^{o(1)}$ factors.} 
\end{itemize}

The easiest way to define length-constrained expander decompositions is in terms of flows; we do so informally here. Given a graph, an \textbf{$h$-length demand} is a function that maps pairs of nodes within distance $h$ to positive integers. A graph is an \textbf{$(h,s)$-length $\phi$-expander} if any $h$-length demand in which each node sends and receives at most its degree can be routed by a multi-commodity flow over length $hs$ paths that sends at most (about) $1/\phi$ flow over any edge. \textbf{An $(h,s)$-length $\phi$-expander decomposition} is a collection of length increases which renders the graph an $(h,s)$-length $\phi$-expander; notice that by increasing the distance between nodes we restrict which demands are $h$-length, making it easier for the graph to be an $(h,s)$-length $\phi$-expander. The above works all use the fact that length-constrained expander decompositions slightly modify the graph so as to control both the amount of flow over edges and the lengths of paths used by flows.  A (non-length-constrained) expander and expander decomposition is defined in the same way but without restrictions on demand pair distances and routing path lengths; see \cite{saranurak2019expander}.

A slightly less intuitive  but equivalent way of defining $(h,s)$-length $\phi$-expanders is as graphs with no sparse length-constrained cuts. In fact, in this work we will use this cut characterization rather than the above flow characterization. The particular sense of a sparse length-constrained cut is given by an \textbf{$(h,s)$-length $\phi$-sparse cut}---very roughly, about $\phi\cdot X$ total length increases that make $X$ disjoint pairs of edges within distance $h$ at least distance $hs$-far for some $X > 0$. In the interest of stating our results quickly, we defer a formal definition of $(h,s)$-length $\phi$-sparse cuts and length-constrained expanders and decompositions in terms of these cuts to \Cref{sec:prelim}.

The important parameters of a length-constrained expander decomposition are the length slack and cut slack. The \textbf{length slack} is $s$, as described above. The length slack quantifies how much longer the paths we send flow over are than the ideal $h$. The \textbf{cut slack} of length-constrained expander decomposition $C$ is $|C|/(\phi m)$. Here, $|C|$ is the size of the decomposition---roughly, the total amount it increases length (see \Cref{def:movingcut}). Any such decomposition has size at least $\phi m$ so cut slack measures how much larger the decomposition is than the ideal. 

Prior work characterized and used the tradeoff between cut and length slack. Namely, \cite{haeupler2023parallel} proved the existence of $(h,s)$-length $\phi$-expander decompositions with length slack $s$ and cut slack $\log n \cdot s n^{O(1/s)}$ using the exponential demand, a technical way of smoothly measuring neighborhood size. \cite{haeupler2023parallel} then used decompositions with this tradeoff to demonstrate that $s$-parallel-greedy graphs have arboricity $\log ^ 3 n \cdot s^3 n^{O(1/s)}$.\footnote{See \Cref{sec:conventions} for a definition of arboricity.} Roughly, $s$-parallel-greedy graphs are built by repeatedly matching nodes at distance at least $s$; see \Cref{fig:pgGraph}. Formally, they are as below.
\begin{restatable}[$s$-Parallel-Greedy Graph]{definition}{sPGGraph}\label{dfn:PGGraph}
    Graph $G = (V,E)$ is $s$-parallel-greedy for $s\ge 2$ iff its edges decompose as $E = M_1 \sqcup M_2 \sqcup \ldots$ where for every $i$ we have:
\begin{enumerate}
    \item \textbf{Matchings:} $M_i$ is a matching on $V$;
    \item \textbf{Far Endpoints:} $d_{G_{i-1}}(u,v)> s$ for each $\{u,v\}\in M_i$;
\end{enumerate}
where $G_{i-1}$ is the graph $(V, \bigcup_{j < i}M_j)$ and $d_{G_{i-1}}$ gives distances in $G_{i-1}$ with length-one edges.
\end{restatable}
\begin{figure}
	\centering
	\begin{subfigure}[b]{0.32\textwidth}
		\centering
		\includegraphics[width=\textwidth,trim=0mm 0mm 0mm 0mm, clip]{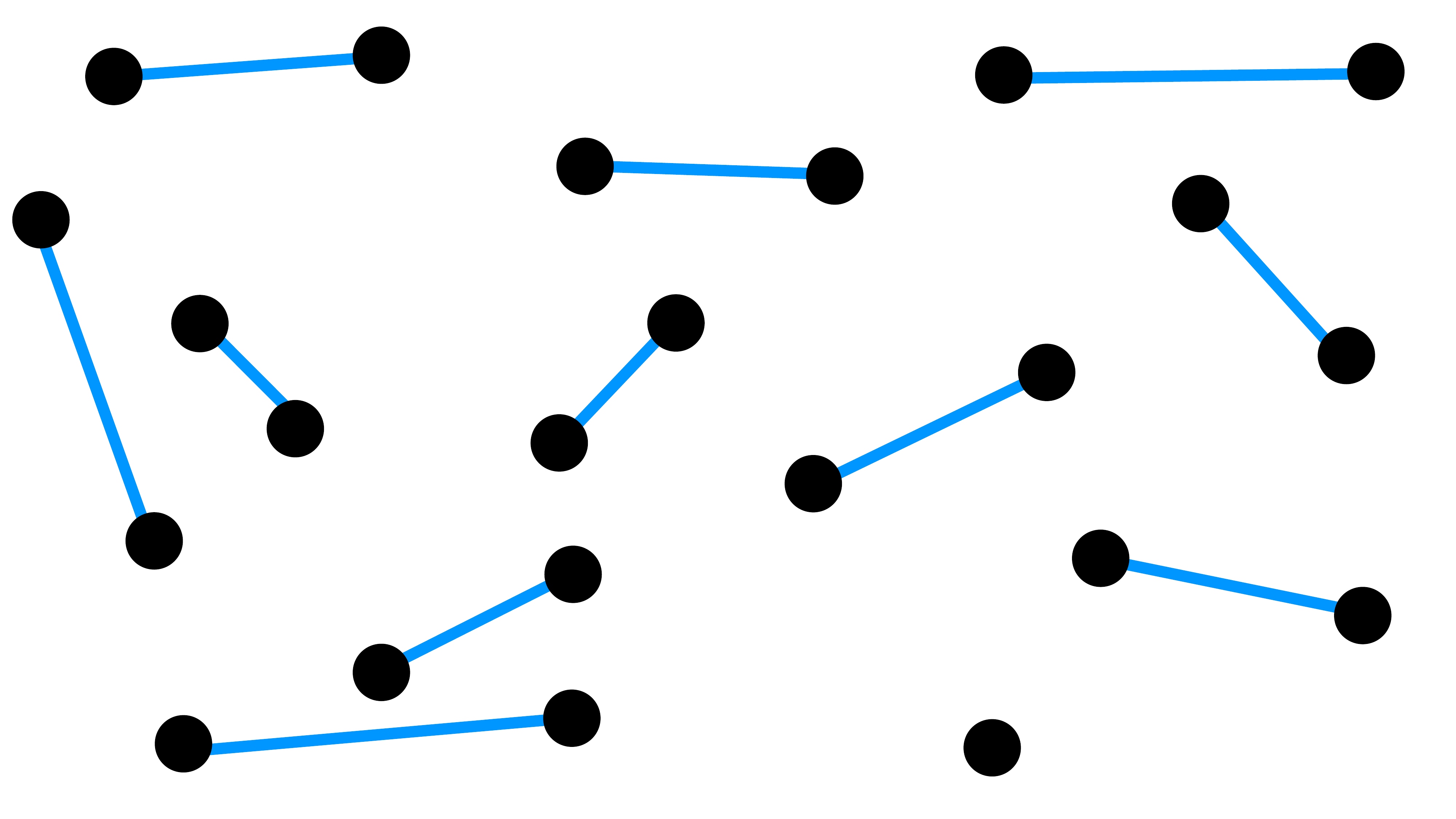}
		\caption{$M_1$}\label{sfig:pg1}
	\end{subfigure}    \hfill
	\begin{subfigure}[b]{0.32\textwidth}
		\centering
		\includegraphics[width=\textwidth,trim=0mm 0mm 0mm 0mm, clip]{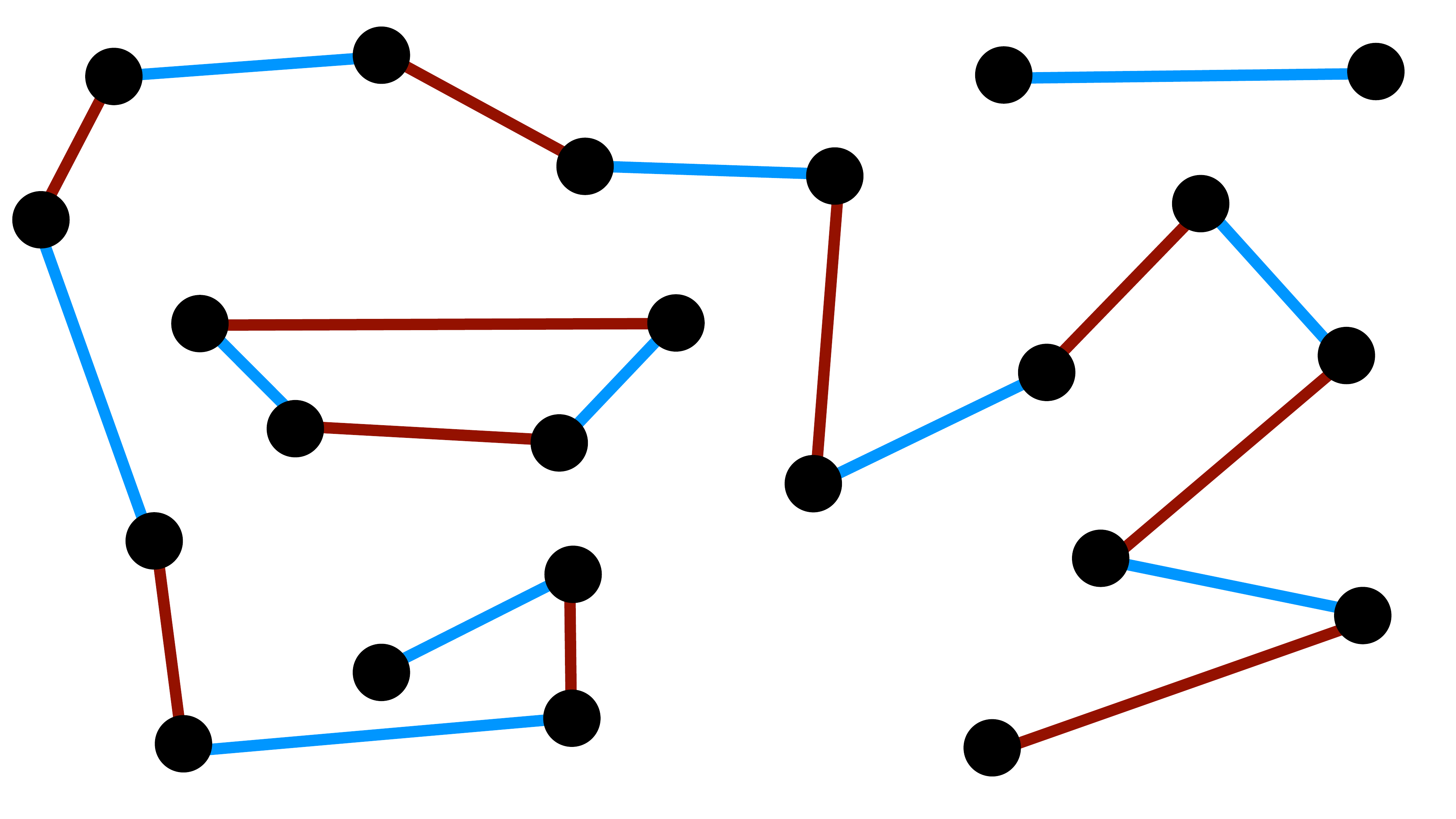}
		\caption{$M_1 \sqcup M_2$.}\label{sfig:pg2}
	\end{subfigure} \hfill
	\begin{subfigure}[b]{0.32\textwidth}
	\centering
	\includegraphics[width=\textwidth,trim=0mm 0mm 0mm 0mm, clip]{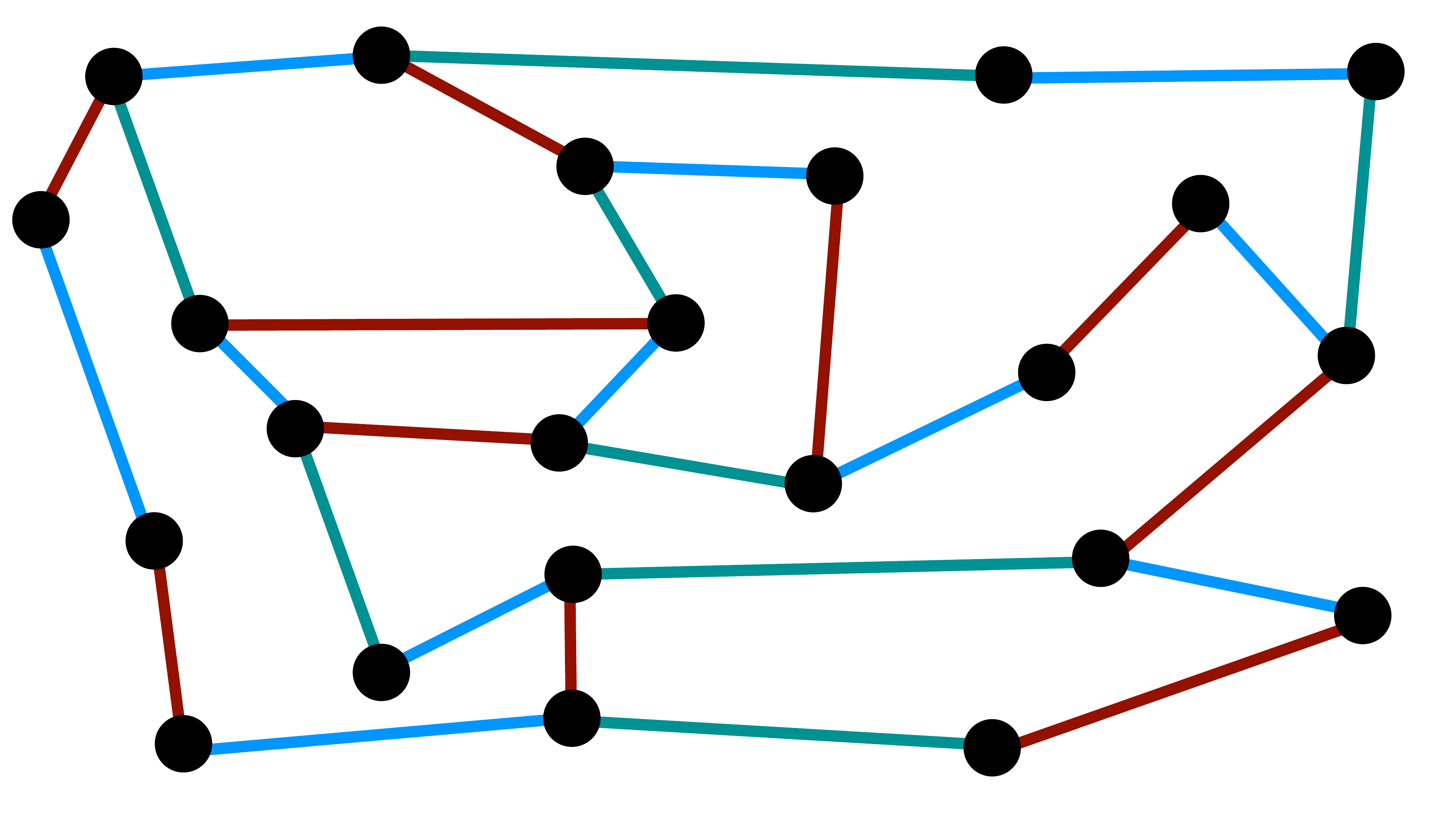}
	\caption{$M_1 \sqcup M_2 \sqcup M_3$.}\label{sfig:pg3}
	\end{subfigure}
	\caption{A $12$-parallel-greedy graph $G = (V, M_1 \sqcup M_2 \sqcup M_3)$.  \ref{sfig:pg3} gives the final graph $G$.}\label{fig:pgGraph}
\end{figure}
\noindent The arboricity bounds of \cite{haeupler2023parallel} are, perhaps, surprising. The graphs output by the greedy algorithm for computing an $s$-spanner is an $s$-parallel-greedy graph where each matching consists of a single edge \cite{althofer1993sparse}. These graphs are known to have arboricity $n^{O(1/s)}$, but typically this is shown by arguing that they have girth at least $s+2$.\footnote{The girth of a graph is the length of its shortest cycle.}
On the other hand, $s$-parallel-greedy graphs can have girth as small as $4$ for any $s$---see, for example, \Cref{sfig:pg3}---and so it is not clear that they should have bounded arboricity.

Later, \cite{haeupler2024new} used these arboricity bounds to make the tradeoff between length and cut slack algorithmic, albeit exponentially-worse than what is existentially-possible. In particular, \cite{haeupler2024new} gave close-to-linear time algorithms achieving length slack $s$ and cut slack $n^{O(1/\poly(\log s))}$.\footnote{We say $f(n) = \poly(n)$ if there exists a constant $c$ such that $f(n) = O(n^c)$.} The key way in which \cite{haeupler2024new} used these arboricity bounds was to show a ``union of length-constrained cuts'' fact: \cite{haeupler2024new} showed that, given a sequence of length-constrained cuts in which each is $(h,s)$-length $\phi$-sparse after applying all previous cuts, the union of these cuts is itself a $\phi'$-sparse $(h',s')$-length cut where $h' \approx h$, $s' \approx s$ and $\phi' \approx  \phi \cdot \alpha$ and $\alpha$ is the arboricity of an $s$-parallel-greedy graph. \cite{haeupler2024new} combined this with the arboricity bound of \cite{haeupler2023parallel} to instantiate this union of length-constrained cuts fact with a $\log ^3 n\cdot s^3 n^{O(1/s)}$ loss in sparsity. This, in turn, gave the foundation of the algorithm of \cite{haeupler2024new}. Similar algorithms based on the union of cuts in the non-length-constrained setting are known where an $O(\log n)$ loss in sparsity is possible \cite{chuzhoy2020deterministic,haeupler2024new}.  

\subsection{Our Contributions}
We give a simple proof of the existence of length-constrained expander decompositions with an improved cut slack of $s n^{O(1/s)}$. We achieve this by improving the union of length-constrained cuts fact by way of improved arboricity bounds for parallel-greedy graphs. Our main result is below.
\begin{restatable}[Existence of Length-Constrained Expander Decompositions---Simplified]{theorem}{LCEDExistSimple}
\label{thm:LCEDsSimple}
    For any $h \geq 1$, $s \geq 2$, $\phi > 0$ and any graph $G = (V,E)$ with edge lengths, there is an $(h,s)$-length $\phi$-expander decomposition  with cut slack $s \cdot n^{O(1/s)}$.
\end{restatable} 
\noindent Our cut slack of $s \cdot n^{O(1/s)}$ improves on the $\log n \cdot s n^{O(1/s)}$ of \cite{haeupler2023parallel}. Also, our proof only assumes $s \geq 2$ whereas \cite{haeupler2023parallel} assumed $s \geq 100$. Likewise, our proof does not use the exponential demand, unlike that of \cite{haeupler2023parallel}.

Rather, our proof observes that the existence of a length-constrained expander decomposition is immediate from the aforementioned union of length-constrained cuts fact; we sketch how we prove the existence of a length-constrained expander decomposition from the union of length-constrained cuts fact now. Consider repeatedly applying $(h,s)$-length $\phi$-sparse cuts until no such cuts remain in the graph. Let $C$ be the union of these cuts and let $C$ be our $(h,s)$-length $\phi$-expander decomposition. The resulting graph is trivially an $(h,s)$-length $\phi$-expander since it contains no $(h,s)$-length $\phi$-sparse cuts. We now bound the cut slack of our decomposition $C$. Letting $h'$, $s'$ and $\phi'$ be the parameters from the above union of cuts fact, it is easy to see that no $(h', s')$-length $\phi'$-sparse cut can have size larger than $\phi'  m$ (since the size of such a cut is $\phi' \cdot X$ for some $X$ which is at most $m$). However, by the union of length-constrained cuts fact, $C$ is an $(h', s')$-length $\phi'$-sparse cut and so has size at most $\phi'm$. Thus, the cut slack of our length-constrained expander decomposition $C$ is $|C|/(\phi m) \leq \phi'/\phi$.

Why are we not immediately done from the above argument? To bound the cut slack, we need $\phi' / \phi$ to be small. However, prior to our work, $\phi'$ in the union of length-constrained cuts fact was only known to be $\log ^3 n\cdot s^3 n^{O(1/s)} \cdot \phi$. In other words, if one applied the previous union of length-constrained cuts fact in the above argument, one would get cut slack $\log ^3 n\cdot s^3 n^{O(1/s)} $ which is both worse than the known cut slack of $\log n \cdot s n^{O(1/s)}$ and certainly worse than the $s n^{O(1/s)}$ for which we are aiming. More importantly, as described above the union of length-constrained cuts fact with $\phi ' = \log ^3 n\cdot s^3 n^{O(1/s)} \cdot \phi$ is known by a bound of $\log ^3 n\cdot s^3 n^{O(1/s)}$ on the arboricity of $s$-parallel-greedy graphs. This arboricity bound, in turn, is known by the existence of length-constrained expander decompositions. As such, if the above argument uses the known union of length-constrained cuts fact, then it proves the existence of length-constrained expander decompositions assuming the existence of length-constrained expander decompositions! Rather, what is needed is an improved bound on the arboricity of $s$-parallel-greedy graphs that does not assume the existence of length-constrained expander decompositions.

We give exactly such a bound on the arboricity of parallel-greedy graphs, as stated below.
\begin{restatable}[Parallel-Greedy Graph Arboricity]{theorem}{PGArb}\label{thm:newPGArb}
If $G$ is an $n$-node $s$-parallel-greedy graph, then $G$ has arboricity at most $O(s \cdot n^{2/s})$.
\end{restatable}

\noindent When $s$ is odd, we actually prove a slightly stronger bound of $O(s \cdot n^{2/(s+1)})$, but this difference does not matter for our applications.
The above improves on the arboricity bound of $\log ^3 n\cdot s^3 n^{O(1/s)}$ of \cite{haeupler2023parallel} and, unlike the proof of \cite{haeupler2023parallel}, does not assume the existence of length-constrained expander decompositions. Instead, it is based on a ``dispersion / counting'' framework used in recent work on graph spanners \cite{Bodwin25, bodwin2024fault}.

Leveraging the connection established by \cite{haeupler2024new} between the arboricity of parallel-greedy graphs and the loss in sparsity of the union of length-constrained cuts, we get the following new union of sparse length-constrained cuts fact. Below, an $(h,s)$-length $\phi$-sparse sequence of length-constrained cuts is one in which each length-constrained cut in the sequence is $(h,s)$-length $\phi$-sparse in the graph after all previous cuts in the sequence have been applied---see \Cref{def:sparsity} for a formal definition. 

\begin{restatable}[Union of Sparse Length-Constrained Cuts---Simplified]{theorem}{uOfCutsSimple}\label{thm:UofCutsSimple}
Fix $h \geq 1$, $s \geq 2$ and $\phi > 0$ and let $(C_1, C_2,\ldots )$ be an $(h,s)$-length $\phi$-sparse sequence of length-constrained cuts in a graph with edge lengths. Then $\left(1+\frac{1}{s-1}\right)\sum_i C_i$ is an $(h',s')$-length $\phi'$-sparse cut with
    \begin{align*}
        h' = 2h \text{\qquad and \qquad} s' = \frac{(s-1)}{2} \text{\qquad and \qquad} \phi' = s  
        n^{O(1/s)} \cdot \phi.
    \end{align*}
\end{restatable}
\noindent This is called a \emph{union} of cuts fact because, since we are dealing with general length increases, taking a union corresponds to taking a sum. Our loss in sparsity of $s n ^{O(1/s)}$ improves on the loss of $\log ^3 n\cdot s^3 n^{O(1/s)}$ from \cite{haeupler2024new}. The values of $h'$, $s'$ and the scaling by $(1+\frac{1}{s-1})$ is identical to \cite{haeupler2024new}. Applying our new union of length-constrained cuts fact with the above argument immediately gives the existence of length-constrained expander decompositions with length slack $s$ and cut slack $s \cdot n^{O(1/s)}$ as formalized in \Cref{thm:LCEDsSimple}.


In fact, our final results on the existence of length-constrained expander decompositions and the union of sparse length-constrained cuts are more general than the above---they apply to general edge-capacitated graphs and arbitrary node-weightings which are, roughly, functions that specify the parts of the graph in which we are interested. 
\Cref{thm:LCEDs} and \Cref{thm:UofCuts} give the more general versions of of \Cref{thm:LCEDsSimple} and
\Cref{thm:UofCutsSimple} respectively.


\section{Preliminaries}\label{sec:conventions}

Before moving on to our results, we introduce the notation and conventions that we use throughout this work and give a more formal definition of length-constrained cuts, expanders and expander decompositions.

\subsection{Graphs and Arboricity}
We make use of the following conventions regarding graphs and their arboricity.
\paragraph*{Graphs Conventions.}
Let $G=(V,E)$ be an undirected graph. We will use $n := |V|$ and $m:= |E|$ for the number of vertices and edges respectively. For $S \subseteq V$, we let $E(S)$ be all edges of $E$ with both endpoints in $S$.

We will associate two functions with the edges of graph $G$. We clarify these here.
\begin{enumerate}
    \item \textbf{Capacities.} We will let $\U = \{\U_e\}_e \in \mathbb{Z}_{\geq 0}^m$ be the capacities of edges of $E$. If no capacities are specified then each edge is treated as having capacity $1$. Informally, these capacities specify the maximum amount of flow which may go over an edge and the cost of cutting an edge. For each vertex $v\in V$, we denote by $\deg_{G}(v) = \sum_{e \ni v} \U_e$ the (capacity-weighted) degree of $v$ in $G$. We assume each $\U_e$ is $\poly(n)$.
    \item \textbf{Lengths.} We will let $ \l = \{\l_e\}_e\in \mathbb{R}_{\geq 0}^m$ be the \emph{lengths} of edges in $E$. These lengths determine the lengths with respect to which we are computing length-constrained expander decompositions. We will let $d_G(u,v)$ give the minimum length of a path in $G$ that connects $u$ and $v$ where the length of a path $P$ is $\l(P) := \sum_{e \in P} \l(e)$. Unlike with capacities, we do not assume that lengths are $\poly(n)$.\footnote{Prior work used length-constrained expanders in the context of length-one edges and talked about $h$-hop expanders and $h$-hop expander decompositions. We use general lengths and use ``length'' instead of ``hop'' where appropriate.}
\end{enumerate}

\paragraph{Graph Arboricity.} A \emph{forest cover} of graph $G$ consists of sub-graphs $F_1, F_2, \ldots, F_k$ of $G$ that are forests such that for $j\neq i$ we have $F_i$ and $F_j$ are edge-disjoint and every edge of $G$ occurs in some $F_i$. $k$ is called the size of the forest cover and the \emph{arboricity} $\alpha$ of $G$ is the minimum size of a forest cover of $G$. We will make use of the following famous Nash-Williams characterization of graph arboricity in terms of ``density.''
\begin{theorem}[\cite{nash1961edge,nash1964decomposition,chen1994short}]\label{lem:NW}
    Graph $G = (V,E)$ has arboricity at most $\alpha$ iff for every $U \subseteq V$ we have $|E(U)| \leq \alpha \cdot (|U|-1)$.
\end{theorem}

\subsection{Background on Length-Constrained Expander Decompositions}\label{sec:prelim}

We next review key definitions from previous work on length-constrained expander decompositions---primarily \cite{haeupler2022expander} and \cite{haeupler2024new}---of which we make use. Readers only interested in the arboricity bounds on parallel-greedy graphs can skip to \Cref{sec:PGBounds}.

Throughout this section, we assume we are given a graph $G$ with edge lengths and capacities as above a length-constraint $h \geq 1$, a length slack $s \geq 2$ and a sparsity $\phi \geq 0$.

\subsubsection{Demands and Node-Weightings}

Demands give a way of specifying how much flow one wants between endpoints in a graph and, more importantly for our purposes, which pairs are separated by a cut.

\begin{definition}[Demand]
A \emph{demand} $D:V\times V\rightarrow\mathbb{Z}_{\ge0}$ is a function that assigns a non-negative value to each ordered pair of vertices. The size of a demand $D$ is $|D| := \sum_{v,w} D(v,w)$. 
A demand $D$ is \emph{$h$-length} if $D(u,v) > 0$ only when $d_G(u,v) \leq h$.
\end{definition}


Node-weightings will give us a useful way of specifying a subset of a graph and, in particular, specifying what it means for a subset of a graph to be a length-constrained expander.

\begin{definition}[Node-Weighting]
    A \emph{node-weighting} is a function $A:V\rightarrow\mathbb{Z}_{\ge0}$ such that for each vertex $v$ we have $A(v) \leq \deg_G(v)$. 
\end{definition}
\noindent Notice that $\deg_G$ is itself a node-weighting. Throughout the rest of this section, we will define several things ``with respect to a node-weighting $A$.'' The same definitions hold not with respect to $A$ when we take $A$ to be $\deg_G$, i.e. when the definition holds with respect to the entire graph. For example, \Cref{def:sparsity} defines what it means for a length-constrained cut to be $(h,s)$-length $\phi$-sparse with respect to node-weighting $A$. Similarly, a cut is said to be simply $(h,s)$-length $\phi$-sparse iff it is $(h,s)$-length $\phi$-sparse with respect to $\deg_G$.

Once we have specified a subset of our graph using a node-weighting, we will be interested in only those demands which ``respect'' this node-weighting, as below.

\begin{definition}[Demand Respecting a Node-Weighting]\label{dfn:demandResp}
Given demand $D$ and node-weighting $A$, we say that $D$ is $A$-respecting if for each vertex $v$ we have $\max\{\sum_{w \in V}D(v,w),\sum_{w \in V}D(w,v)\} \leq A(v)$. 
\end{definition}

\subsubsection{Length-Constrained Cuts and Length-Constrained Cut Sparsity} We now define length-constrained cuts and what it means for a length-constrained cut to be sparse by using the above notion of demands and node-weightings.

The following is the sense of a length-constrained cut that we will use. In particular, we must use a notion of cuts which fractionally cuts edges because in the length-constrained setting, there are known large flow-cut gaps in the integral case \cite{baier2010length}. Here, a cut value of $1$ corresponds to fully cutting an edge. 
\begin{definition}[Length-Constrained Cut]\label{def:movingcut}
A length-constrained cut is a non-zero function $C: E \mapsto \mathbb{R}_{\geq 0}$. The \emph{size} of $C$ is  $|C|:=\sum_{e} \U_e \cdot C(e)$.
    
\end{definition}


\noindent If we interpret length-constrained cuts as inducing length increases, then we can formalize applying them in a graph as follows. The intuition behind this sense of applying a cut is: when we are working with a length constraint $h$, fully cutting an edge corresponds to increasing its length to $h$ since no path with length at most $h$ could ever use such an edge.
\begin{definition}[$G-C_h$]
Given $h \geq 1$ and graph $G$ with length function $l$, we let $G-C_h$ be the graph $G$ with length function $l +h \cdot C$.
\end{definition}
\noindent We will refer to $G-C_h$ as the result of \emph{applying $C$} to graph $G$.

Using this sense of applying a cut, we get the following length-constrained analogue of disconnecting two vertices, namely making them $h$-far.
\begin{definition}[$h$-Separation]
Given length-constrained cut $C$ and $h \geq 1$, we say nodes $u,v \in V$ are $h$-separated by $C$ if their distance in $G-C_h$ is strictly larger than $h$, i.e.\ $d_{G-C_h}(u,v) > h$.
\end{definition}
\noindent Similarly, we define how much a cut separates a demand.

\begin{definition}[$h$-Separated Demand]\label{dfn:sepDem}
We define the amount of demand $D$ separated by length-constrained cut $C$ as 
\begin{align*}
    \sep_{h}(C,D) := \sum_{(u,v): \text{ } h\text{-separated by } C} D(u,v),
\end{align*}
the sum of demand between vertices that are $h$-separated by $C$.
\end{definition}
\noindent If a demand $D$ and length-constrained cut $C$ are such that $\sep_h(C,D) = |D|$ then we will simply say that $C$ $h$-separates $D$.

$h$-separating an $h$-length demand is, generally speaking, too easily achieved. In particular, if each pair were at distance exactly $h$ then $h$-separating would only require making endpoints some tiny distance further. Rather, generally we will be interested in $hs$-separating an $h$ demand for some length slack $s \geq 2$. With this in mind, we can define the following notion of the demand-size of a length-constrained cut. Generally speaking, this is the length-constrained analogue of the ``volume'' of a cut---see \cite{haeupler2024new}. Likewise, here the node-weighting specifies the parts of the graph with respect to which we are measuring volume.

\begin{definition}[$(h,s)$-Length Demand-Size]\label{def:demandSize}
    Given length-constrained cut $C$ and node-weighting $A$, we define the $(h,s)$-length demand-size $A_{(h,s)}(C)$ of $C$ with respect to $A$ as the size of the largest $A$-respecting $h$-length demand which is $hs$-separated by $C$. That is,
    \begin{align*}
        A_{(h,s)}(C) := \max \left\{|D| : \text{$D$ is an $A$-respecting $h$-length demand with $\sep_{hs}(C,D) = |D|$}\right\}.
    \end{align*}
\end{definition}




\noindent We refer to the maximizing demand $D$ above as the \emph{witness} demand of $C$ with respect to $A$.

Having defined the length-constrained analogue of the volume of a cut, we can now define length-constrained sparsity which, analogous to the classic setting, is the size of a cut divide by its volume (or, in this case, the demand-size).

\begin{definition}[$(h,s)$-Length Sparsity of a Cut]\label{def:sparsity}
The $(h,s)$-length sparsity of length-constrained cut $C$ with respect to a node-weighting $A$ is
\begin{align*}
    \spa_{(h,s)}(C,A) = \frac{|C|}{A_{(h,s)}(C)},
\end{align*}
 the cut size divided by the demand-size.
\end{definition}
\noindent We will say that a cut $C$ is $(h,s)$-length $\phi$-sparse with respect to $A$ if $\spa_{(h,s)}(C,A) \leq \phi$.

We will be interested in applying a sequence of length-constrained cuts which are sparse after we have applied all previous cuts. The following formalizes this.
\begin{definition}[Sequence of Length-Constrained Cuts]\label{dfn:cutSeq}
Given graph $G = (V,E)$ and node-weighting $A$, a sequence of length-constrained cuts $(C_1, C_2, \ldots)$ is $(h,s)$-length $\phi$-sparse with respect to $A$ if each $C_i$ is $(h,s)$-length $\phi$-sparse in $G - \left(\sum_{j < i} C_j\right)_{hs}$.
\end{definition}
\noindent We will say that an $(h,s)$-length $\phi$-sparse sequence $(C_1, C_2, \ldots)$ is maximal if $G - \left(\sum_{i} C_i\right)_{hs}$ does not contain any $(h,s)$-length $\phi$-sparse cuts.


\subsubsection{Length-Constrained Expanders and Length-Constrained Expander Decompositions}\label{sec:h-length-expander-def}
We can now define a length-constrained expander as a graph which does not contain any sparse length-constrained cuts.

\begin{definition}[$(h,s)$-Length $\phi$-Expander] A graph $G$ with edge lengths and capacities is an $(h,s)$-length $\phi$-expander with respect to node-weighting $A$ if there does not exists a length-constrained cut that is $(h,s)$-length $\phi$-sparse with respect to $A$ in $G$.
\end{definition}


A length-constrained expander decomposition is simply a length-constrained cut which renders the graph a length-constrained expander. 

\begin{definition}[$(h,s)$-Length $\phi$-Expander Decomposition] \label{def:LCED}
Given graph $G$, an \emph{$(h,s)$-length $\phi$-expander decomposition} with respect to node-weighting $A$ with cut slack $\kappa$ and length slack $s$ is a length-constrained cut $C$ of size at most $\kappa \cdot \phi|A|$ such that $G-C_{hs}$ is an $(h,s)$-length $\phi$-expander with respect to $A$.
\end{definition}

\section{New Arboricity Bounds for Parallel-Greedy Graphs}\label{sec:PGBounds}
We now prove our new bounds on the arboricity of $s$-parallel-greedy graphs, redefined below.
\sPGGraph*
\noindent The below gives our new arboricity bound for parallel-greedy graphs, this section's main result.
\PGArb*
\noindent At a high level, we achieve our new bounds by applying a ``dispersion/counting'' proof framework to paths which are ``monotonic'' with respect to the parallel-greedy graph's matchings and which consist of $s/2$ edges. For this framework, we argue (1) a dispersion lemma which upper bounds  the number of such paths by arguing that they must be ``dispersed'' around the graph, rather than concentrated between two vertices and (2) a counting lemma which lower bounds the number of such paths in terms of the average degree. Combining these lemmas upper bounds the average degree which, since any subgraph of an $s$-parallel-greedy graph is an $s$-parallel-greedy graph, serves to upper bound the arboricity by \Cref{lem:NW}. A similar framework has been used in recent work on graph spanners; for example, \cite{Bodwin25, bodwin2024fault} use this framework over related (but more specific) types of paths. 

For the rest of this section we assume we are given an $n$-node $s$-parallel greedy graph $G = (V,E)$ with $m$ edges whose arboricity we aim to bound.
Likewise, we let $(M_1, \dots, M_k)$ be an ordered sequence of matchings that partition the edge set $E$, witnessing $G$ is an $s$-parallel-greedy graph. Also, for the rest of this section, we refer to a path with exactly $s/2$ edges as an $\frac{s}{2}$-path and for simplicity of presentation we assume that $s$ is even; in the case where $s$ is odd, the same proof works with respect to $\frac{s+1}{2}$-paths (leading to the slightly-improved bound of $O(s \cdot n^{2/(s+1)})$ mentioned previously).

The following formalizes the sense of monotonic paths we use.
\begin{definition} [Monotonic Paths]
	A path $P$ in $G$ is monotonic if the edges in $P$ occur in exactly the
	same order as the matchings that contain these edges.
	In other words, let $(e_1, e_2, \dots, e_x)$ be the edge sequence of $P$, and let $M_{i_j}$ be the matching that contains edge $e_j$ for each $1\le j\le x$.
	Then we say that $P$ is monotonic if we have $i_1 < i_2 < \dots < i_x$.
\end{definition}
\noindent The rest of this section proves \Cref{thm:newPGArb} by counting the number of monotonic $\frac{s}{2}$-paths.

\subsection{Dispersion Lemma}
Our dispersion lemma shows that monotonic $\frac{s}{2}$-paths must be ``dispersed'' around the graph, rather than be  concentrated on one pair of endpoints. This lemma will use a slightly different characterization of $s$-parallel-greedy graphs as below.
\begin{lemma}\label{lem:PGAlt}
    For any cycle $C$ of $s$-parallel-greedy graph $G$ with $|C| \le s+1$ edges, if $M_i$ is the highest-indexed matching that contains an edge of $C$, then there are least two edges from $M_i$ in $C$.
\end{lemma}
\begin{proof}
Suppose for the sake of contradiction that $C$ only contained one edge $\{u,v\}$ from $M_i$. Then, $G_{i-1}:= (V,\bigcup_{j < i} M_i)$ contains every edge other than $\{u,v\}$ of $C$ of which there are at most $s$ so 
\begin{align}\label{eq:sagsa}
    d_{G_{i-1}}(u,v) \leq s.
\end{align}
 But, $\{u, v\} \in M_i$ and $G$ is $s$-parallel-greedy so $d_{G_{i-1}}(u,v) > s$ which contradicts \Cref{eq:sagsa}.
\end{proof}
\noindent See \Cref{sfig:pg3} for an illustration of this on a $12$-parallel-greedy graph; in this graph, there are many cycles with at most $13$ edges but each such cycle has at least two edge from its highest-indexed incident matching.

The following is our dispersion lemma.
\begin{lemma} [Dispersion Lemma] \label{lem:disp}
For $u, v \in V$, there is at most one monotonic $\frac{s}{2}$-path from $u$ to $v$ in $G$.
\end{lemma}
\begin{proof}
Suppose for contradiction that there are two distinct $\frac{s}{2}$-paths from $u$ to $v$ , $P_a$ and $P_b$; see \Cref{sfig:dispLem1}.
Then there exist contiguous subpaths $Q_a \subseteq P_a, Q_b \subseteq P_b$ such that $Q_a \cup Q_b$ forms a cycle $C$.
Note that the number of edges in $C$ satisfies
$$|C| \le |Q_a| + |Q_b| \le |P_a| + |P_b| = s,$$
and so by \Cref{lem:PGAlt}, we know that the highest-indexed matching containing an edge of $C$ must contain at least $2$ edges of $C$. We proceed to contradict this.

Let $e_a^*, e_b^*$ be the last edges of $Q_a, Q_b$ respectively; see \Cref{sfig:dispLem1}.
These edges share an endpoint (since they are adjacent in C), and therefore they belong to different matchings.
We will assume without loss of generality that $e_a^*$ is in a higher-indexed matching than $e_b^*$.
Since $Q_a, Q_b$ are monotonic paths, this implies that \emph{every} edge in $Q_b$ is in a lower-indexed matching than the one containing $e_a^*$, and also every edge in $Q_a$ (except $e_a^*$ itself) comes from a lower-indexed matching than the one containing $e_a^*$; see \Cref{sfig:dispLem3}.
Thus the highest-indexed matching to contribute an edge to $C$ only contributes one edge, $e_a^*$, a contradiction.
\end{proof}
\noindent Notice that it follows by the dispersion lemma that there are $O(n^2)$ monotonic $\frac{s}{2}$-paths in $G$.

\begin{figure}
	\centering
	\begin{subfigure}[b]{0.32\textwidth}
		\centering
		\includegraphics[width=\textwidth,trim=0mm 50mm 0mm 50mm, clip]{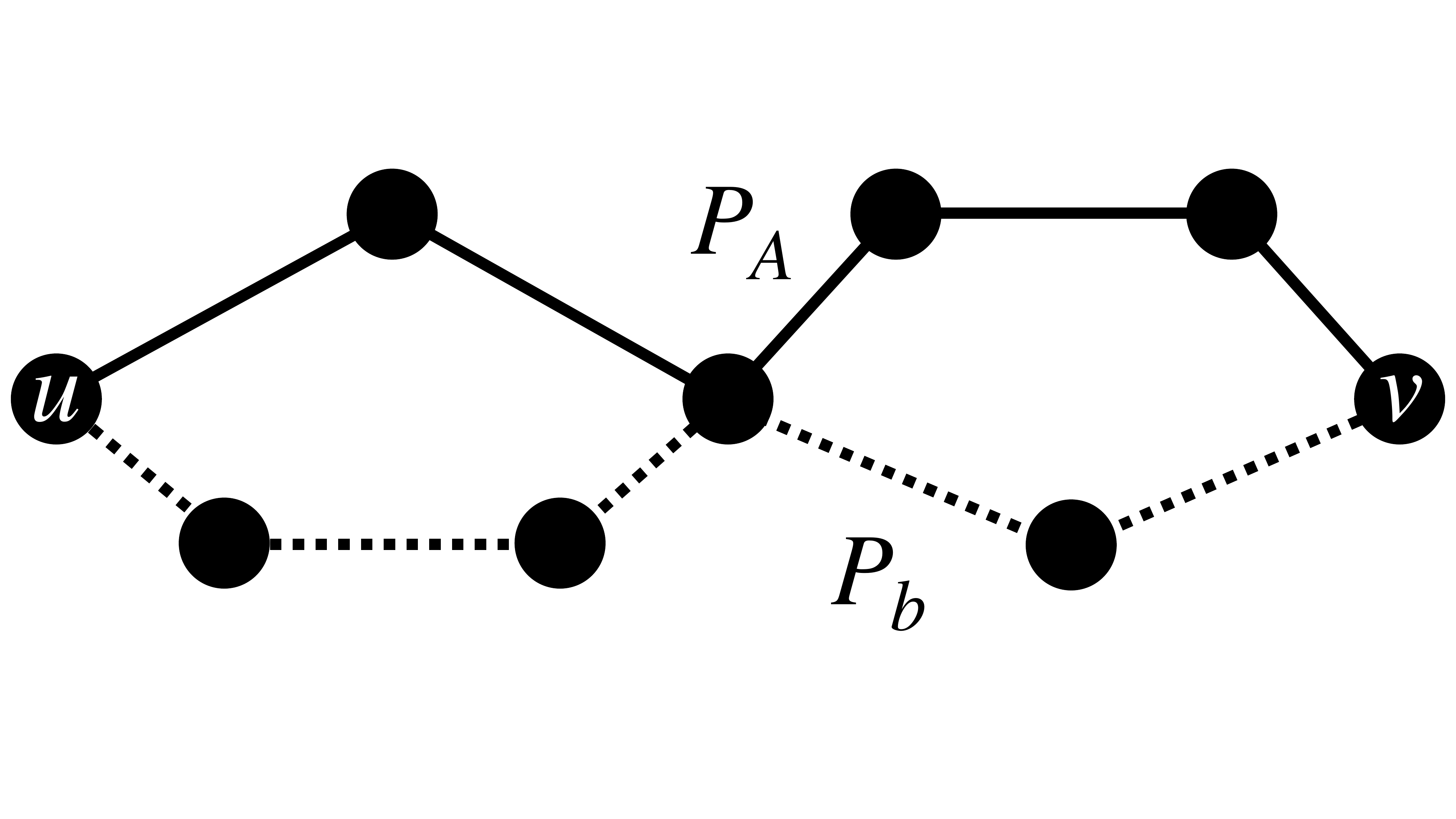}
		\caption{$P_a$, $P_b$.}\label{sfig:dispLem1}
	\end{subfigure}    \hfill
	\begin{subfigure}[b]{0.32\textwidth}
		\centering
		\includegraphics[width=\textwidth,trim=0mm 50mm 0mm 50mm, clip]{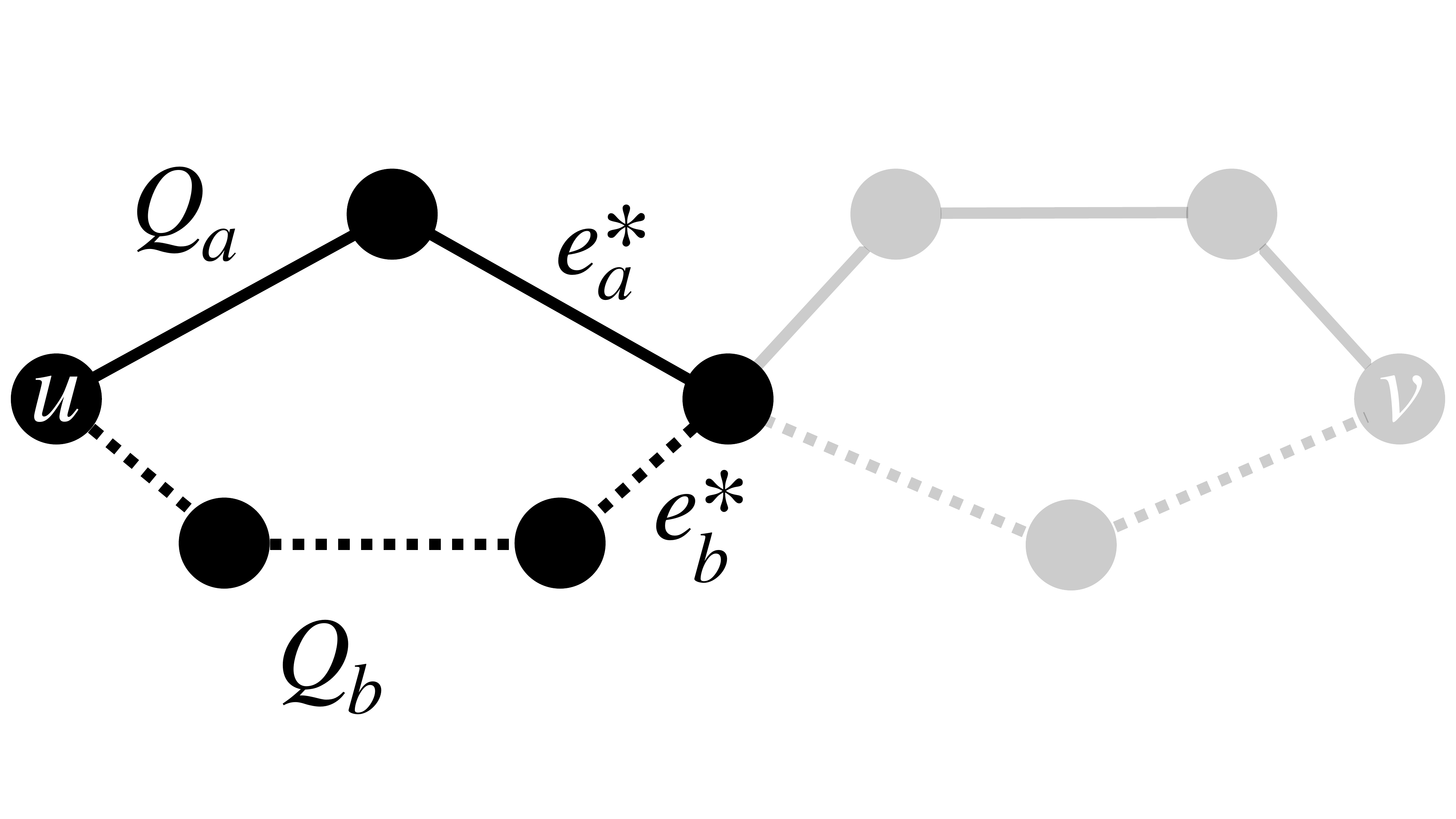}
		\caption{$Q_a$, $Q_b$.}\label{sfig:dispLem2}
	\end{subfigure} \hfill
	\begin{subfigure}[b]{0.32\textwidth}
	\centering
	\includegraphics[width=\textwidth,trim=0mm 50mm 0mm 50mm, clip]{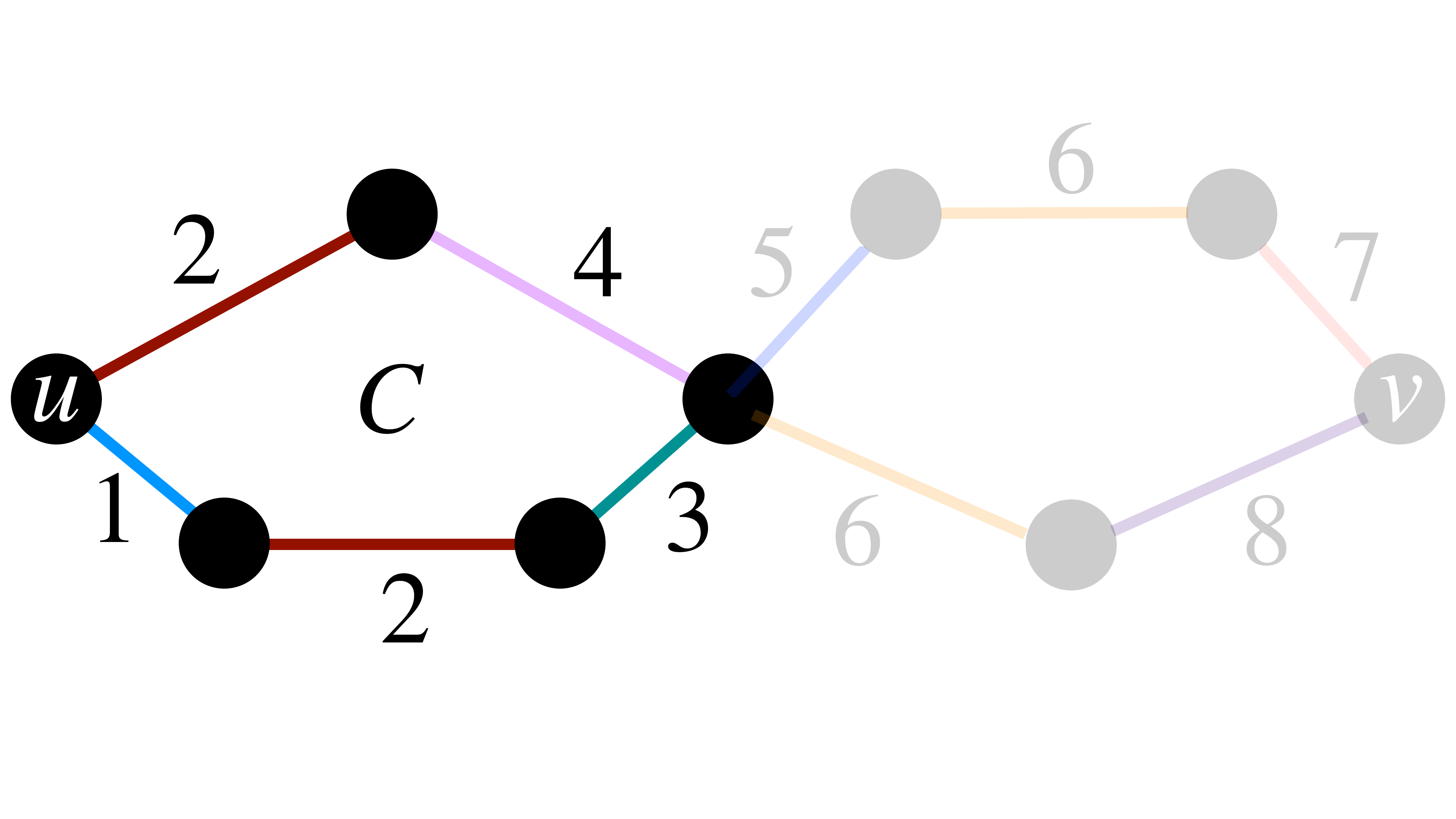}
	\caption{$C$ and matching indices.}\label{sfig:dispLem3}
	\end{subfigure}
	\caption{The proof of our dispersion lemma (\Cref{lem:disp}) for $s=10$. \ref{sfig:dispLem1} gives paths $P_a$ (solid) and $P_b$ (dashed) from $u$ to $v$. \ref{sfig:dispLem2} gives $Q_a$, $Q_b$, $e_a^*$ and $e_b^*$. \ref{sfig:dispLem3} labels each edge with the index of its matching and cycle $C$ which contradicts \Cref{lem:PGAlt}.}\label{fig:dispLem}
\end{figure}

\subsection{Counting Lemma}
Next we will prove our counting lemma, giving a lower bound on the number of monotonic $\frac{s}{2}$-paths in $G$ in terms of the average degree. We bootstrap up to our ``full'' counting lemma using a ``weak'' and ``medium'' counting'' lemma. Ultimately, however, we will only need the full counting lemma.

Our weak counting lemma follows from an adaptation of the folklore \emph{hiker lemma} in graph theory; see e.g.~\cite{Bodwin25, bodwin2024fault} for discussion and similar usage.

\begin{lemma} [Weak Counting Lemma]\label{lem:weakCount}
If $m \ge sn/4$, then $G$ contains a monotonic $\frac{s}{2}$-path.
\end{lemma}
\begin{proof}

Start by placing a ``hiker'' at each node of $G$.
Then, consider the matchings $M_i$ in increasing order of index.
Upon matching $M_i$, we ask the hikers who currently stand at the endpoints of each edge $(x, y) \in M_i$ to walk across that edge, switching places with each other.
Once all matchings have been considered, our $n$ hikers have traversed $2m \ge sn/2$ edges in total.
Additionally, each individual hiker has walked a monotonic path.
Thus, there exists a hiker who walked a monotonic path of $\ge s/2$ edges (and so any subpath of exactly $s/2$ edges satisfies the lemma).
\end{proof}
\noindent We next give the medium version of our counting lemma using our weak lemma.
\begin{lemma} [Medium Counting Lemma]\label{lem:medCount}
If $m \ge sn/2$, then $G$ has at least $\Omega(n)$ monotonic $\frac{s}{2}$-paths.
\end{lemma}
\begin{proof}
By the weak counting lemma (\Cref{lem:weakCount}), we would need to delete at least $sn/4 \ge \Omega(n)$ edges in order to make $G$ have no monotonic $\frac{s}{2}$-paths.
This implies that $G$ has at least $\Omega(n)$ monotonic $\frac{s}{2}$-paths, since otherwise we could delete one edge from each path to eliminate them all.
\end{proof}
\noindent Lastly, we prove our full counting lemma by applying our medium counting lemma.
\begin{lemma} [Full Counting Lemma] \label{lem:fullcount}
Let $d$ be the average degree in $G$.
If $d \ge s$, then $G$ contains at least $n \cdot \Omega(d/s)^{s/2}$ monotonic $\frac{s}{2}$-paths.
\end{lemma}
\begin{proof}
Let $G'$ be a uniform random edge-subgraph of $G$ on exactly $sn/2$ edges.
Let $x$ be the number of monotonic $\frac{s}{2}$-paths in $G$, and let $x'$ be the number of monotonic $\frac{s}{2}$-paths that survive in $G'$.
On one hand, by the medium counting lemma (\Cref{lem:medCount}), we have $x' \ge \Omega(n)$ (deterministically).
On the other hand, for any monotonic $\frac{s}{2}$-path $P$ in $G$, the probability that $P$ survives in $G'$ is
\begin{align*}
&\underbrace{\frac{sn/2}{m}}_{\text{probability first edge is selected in $G'$}} \cdot \underbrace{\frac{sn/2-1}{m-1}}_{\substack{\text{probability second edge is selected in $G'$,}\\\text{given first edge is selected in $G'$}}} \cdot \ldots \cdot  \underbrace{\frac{sn/2-(s/2-1)}{m-(s/2-1)}}_{\substack{\text{probability $s/2^{th}$ edge is selected in $G'$,}\\\text{given first $s/2-1$ edges are selected in $G'$}}}\\
\end{align*}
which is
\begin{align*}
\le \ & \left(\frac{sn}{2m}\right)^{s/2}\\
= \ & O\left( \frac{s}{d} \right)^{s/2}.
\end{align*}
Thus we have
$$\Omega(n) \le \mathbb{E}[x'] \le x \cdot O\left( \frac{s}{d} \right)^{s/2}.$$
Rearranging, we get
$$x \ge n \cdot \Omega\left(\frac{d}{s}\right)^{s/2},$$
as claimed.
\end{proof}

\subsection{Completing Our Arboricity Bound}

We now complete our bound on the arboricity of $s$-parallel-greedy graphs by combining our dispersion lemma and full counting lemma.
\PGArb*
\begin{proof}
First, we claim that any $n$-node $s$-parallel greedy graph $G$ has average degree at most $O(s \cdot n^{2/s})$. Let $d$ be the average degree of $G$. By \Cref{lem:disp}, there are $O(n^2)$ monotonic $\frac{s}{2}$-paths in $G$.
By \Cref{lem:fullcount}, there are $n \cdot \Omega(d/s)^{s/2}$ monotonic $\frac{s}{2}$-paths in $G$.
Comparing these estimates, we have
\begin{align*}
    n \cdot \Omega\left(\frac{d}{s}\right)^{s/2} \le O(n^2).
\end{align*}
Rearranging this inequality gives
$$d \le O(s \cdot n^{2/s}),$$
giving our claimed bound on the average degree of $G$.

To bound the arboricity of $G$, observe that any subgraph of an $s$-parallel greedy graph is itself an $s$-parallel greedy graph. Combining this with our bound on the average degree of an $s$-parallel greedy graph, we get that for any $U \subseteq V$ we have 
\begin{align*}
    |E(U)| \leq O(s \cdot |U|^{2/s}) \cdot \left(|U|-1 \right) \leq O(s \cdot n^{2/s})\cdot \left(|U|-1 \right).
\end{align*}
Applying \Cref{lem:NW}, we get that the arboricity of $G$ is at most $O(s \cdot n^{2/s})$ as required.
\end{proof}

\section{The Union of Sparse Length-Constrained Cuts is Sparse}

We next give our new bound on the sparsity of the union of length-constrained cuts, stated in full generality below. Note that the below does not explicitly require that each $C_i$ is an $(h,s)$-length $\phi$-sparse cut.

\begin{restatable}{theorem}{uOfCuts}\label{thm:UofCuts}
Fix $h \geq 1$, $s \geq 2$ and $\phi > 0$ and let $A$ be a node-weighting and $(C_1, C_2,\ldots )$ be length-constrained cuts in a graph with edge lengths and capacities. Then $\left(1+\frac{1}{s-1}\right)\sum_i C_i$ is an $(h',s')$-length $\phi'$-sparse cut with respect to $A$ with
    \begin{align*}
        h' = 2h \text{\qquad and \qquad} s' = \frac{(s-1)}{2} \text{\qquad and \qquad} \phi' = O\left(s \cdot 
        |A|^{2/s} \cdot \frac{\sum_i |C_i|}{\sum_i A_{(h,s)}(C_i)} \right).
    \end{align*}
    where $A_{(h,s)}(C_i)$ is the demand-size of $C_i$ in $G - \left(\sum_{j < i} C_j\right)_{hs}$ .
\end{restatable}
\noindent The earlier-stated simplified version of the above---\Cref{thm:UofCutsSimple}---follows immediately. In particular, if we apply the above fact to an $(h,s)$-length $\phi$-sparse sequence of length-constrained cuts $(C_1, C_2, \ldots)$ in a graph with capacity one on every edge where $A = \deg_G$, then $|A|^{2/s}$ becomes $n^{O(1/s)}$ and for each $i$ we know $|C_i| / A_{(h,s)}(C_i) \leq \phi$ so $\sum_i \frac{|C_i|}{A_{(h,s)}(C_i)} \leq \phi$. It follows that in this case $\phi' =s n^{O(1/s)} \cdot \phi$ as needed for \Cref{thm:UofCutsSimple}.

\Cref{thm:UofCuts}, in turn, follows immediately by our new arboricity bound and a connection established between the arboricity of $s$-parallel-greedy graphs and the sparsity of the union of length-constrained cuts. We summarize this connection below.
\begin{restatable}[\cite{haeupler2024new}]{lemma}{UCviaPGArb}\label{lem:UCviaPGArb}
    Suppose every $n$-node $s$-parallel-greedy graph has arboricity at most $\alpha(n, s)$ and fix $h \geq 1$ and $s \geq 2$. Then if $(C_1, C_2,\ldots )$ is a sequence of length-constrained cuts we have that $\left(1+\frac{1}{s-1}\right)\sum_i C_i$ is an $(h',s')$-length $\phi'$-sparse cut with respect to $A$ with
    \begin{align*}
        h' = 2h \text{\qquad and \qquad} s' = \frac{s-1}{2} \text{\qquad and \qquad} \phi' = 8\alpha\left(|A|,s\right) \cdot \frac{\sum_i |C_i|}{\sum_i A_{(h,s)}(C_i)}.
    \end{align*}
    where $A_{(h,s)}(C_i)$ is the demand-size of $C_i$ in $G - \left(\sum_{j < i} C_j\right)_{hs}$.
\end{restatable}
\noindent For completeness, we give a proof of \Cref{lem:UCviaPGArb} in \Cref{sec:sparseCutsVsPGArb}. Combining our arboricity bound (\Cref{thm:newPGArb}) and \Cref{lem:UCviaPGArb} immediately proves \Cref{thm:UofCuts}.

\section{Existence of Length-Constrained Expander Decompositions}
In this section we prove our main result, the existence of length-constrained expander decompositions with cut slack $O\left(s \cdot |A|^{2/s} \right)$. We state this result in full generality below.
\begin{restatable}[Existence of Length-Constrained Expander Decompositions]{theorem}{LCEDExist}
\label{thm:LCEDs}
    For any $h \geq 1$, $s \geq 2$, $\phi > 0$ and any graph $G = (V,E)$ with edge lengths and capacities and node weighting $A$, there is an $(h,s)$-length $\phi$-expander decomposition with respect to $A$ with cut slack $O\left(s \cdot |A|^{2/s} \right)$.
\end{restatable} 
\noindent Note that the earlier-stated simplified version of the above theorem---\Cref{thm:LCEDsSimple}---follows immediately from the above theorem by using capacity $1$ on all edges and taking $A$ to be the node-weighting $\deg_G$. As described earlier, we prove this by simply repeatedly cutting $(h,s)$-length $\phi$-sparse cuts and then applying our union of length-constrained cuts fact.

We begin by unpacking definitions to formalize the idea that no length-constrained $\phi$-sparse cut can have size larger than $\phi \cdot |A|$ when working with respect to node-weighting $A$. Note that if we had capacity one on every edge and were working with respect to the node-weighting $\deg_G$, this lemma would simply state than an $(h,s)$-length $\phi$-sparse cut has size at most $\phi m$.
\begin{lemma}\label{lem:limCutSize}
    Given graph $G = (V,E)$ with edge lengths and capacities, $h \geq 1$, $s \geq 2$ and $\phi > 0$ and node-weighting $A$, any $(h,s)$-length $\phi$-sparse cut with respect to $A$ has size at most $\phi \cdot |A|$.
\end{lemma}
\begin{proof}
	The proof follows immediately from the relevant definitions. Let $C$ be an $(h,s)$-length $\phi$-sparse cut with respect to $A$. Our goal is to show $|C| \leq \phi \cdot |A|$. By the definition of an $(h,s)$-length $\phi$-sparse cut with respect to $A$ (\Cref{def:sparsity}) , we know that 
	\begin{align}\label{eq:teds}
		\spa_{(h,s)}(C,A) = \frac{|C|}{A_{(h,s)}(C)} = \phi.
	\end{align}
	
	However, notice that by the definition of the $(h,s)$-length demand-size of $C$ (\Cref{def:demandSize}), we know $A_{(h,s)}(C)$ is the size of the largest $h$-length $A$-respecting demand which is $hs$-separated by $C$. However, by definition of an $A$-respecting demand (\Cref{dfn:demandResp}), any demand which is $A$-respecting has size at most $|A|$ and so 
    \begin{align}\label{eq:preb}
        A_{(h,s)}(C) \leq |A|
    \end{align}

    Combining \Cref{eq:teds} and \Cref{eq:preb} we get
	\begin{align*}
		|C| =  \phi \cdot A_{(h,s)}(C) \leq \phi \cdot |A|
	\end{align*}
	as required.
\end{proof}

We now complete our proof of the existence of length-constrained expander decompositions using our new union of length-constrained cuts fact. See \Cref{def:LCED} for a reminder of the definition of a length-constrained expander decomposition.
\LCEDExist*
\begin{proof}
    Let $(C_1, C_2, \ldots)$ be an $(h,s)$-length $\phi$-sparse sequence with respect to $A$ that is maximal and let $C := (1+\frac{1}{s-1})\sum_i C_i$ be the ``union'' of this sequence. We claim that $C$ is an $(h,s)$-length $\phi$-expander decomposition with respect to $A$ with cut slack $O(s \cdot |A|^{2/s})$ and length slack $s$. $G - C_{hs}$ must be an $(h,s)$-length $\phi$-expander with respect to $A$ since $(C_1, C_2, \ldots)$ is maximal. 

     It remains only to show that the cut slack of $C$ is at most $O\left(s \cdot |A|^{2/s}\right)$ for which it suffices to bound the size of $C$ as at most $O(s \cdot |A|^{2/s} \cdot \phi |A|)$. Applying the fact that the union of sparse length-constrained cuts is a sparse length-constrained cut (\Cref{thm:UofCuts}), we have that $C$ is an $(h',s')$-length $\phi'$-sparse cut with respect to $A$ with 
     \begin{align*}
        h' = 2h \text{\qquad and \qquad} s' = \frac{(s-1)}{2} \text{\qquad and \qquad} \phi' = O\left(s \cdot |A|^{2/s} \cdot \frac{\sum_i |C_i|}{\sum_i A_{(h,s)}(C_i)} \right).
    \end{align*}
    where $A_{(h,s)}(C_i)$ is the demand-size of $C_i$ in $G - \left(\sum_{j < i} C_j\right)_{hs}$ .

    Since $(C_1, C_2, \ldots)$ is  $(h,s)$-length $\phi$-sparse with respect to $A$, we get $\frac{|C_i|}{A_{(h,s)}(C_i)} \leq \phi$ for each $i$ so 
    \begin{align*}
        \sum_i \frac{|C_i|}{A_{(h,s)}(C_i)} \leq \phi.
    \end{align*}
    meaning
    \begin{align*}
        \phi' = O(s \cdot |A|^{2/s} \cdot \phi).
    \end{align*}
    However, by \Cref{lem:limCutSize}, we know that $|C| \leq \phi' \cdot |A|$ and so $|C| \leq O(s \cdot |A|^{2/s} \cdot \phi |A|)$ as required.
\end{proof}

\appendix
\section{Union of Length-Constrained Cuts via Parallel Greedy Arboricity}\label{sec:sparseCutsVsPGArb}

In this section we prove that a bound on the arboricity of $s$-parallel-greedy graphs also bounds the loss in sparsity when taking the union of sparse length-constrained cuts. The following summarizes this.

\UCviaPGArb*
\noindent We give a proof of the above for completeness. Most of the proofs (and figure) from this section are taken from \cite{haeupler2024new} but modified to match our notation and conventions where appropriate.

At a high level, the proof works as follows. Given our sequence of sparse length-constrained cuts, we consider the witness demands of these cuts. Next, we observe that these witness demands induce a particular graph---the demand matching graph---which is $s$-parallel-greedy. We then use a forest cover of this graph to define a new demand---the matching-dispersed demand---which is nearly as large as sum of all of the witness demands and which is separated by the union of our sequence of sparse length-constrained cuts. Such a demand witnesses that the union of our sequence of sparse length-constrained cuts is nearly as sparse as the individual cuts.

\subsection{The Demand Matching Graph}

We begin by formalizing the graph induced by the witness demands of our cut sequence. We call this graph the demand matching graph. Informally, this graph simply creates $A(v)$ copies for each vertex $v$ and then matches copies to one another in accordance with the witness demands.
\begin{definition}[Demand Matching Graph]\label{def:demandMatching}
Given a graph $G = (V,E)$, a node-weighting $A$ and $A$-respecting demands $\mcD = (D_1, D_2, \ldots )$ we define the demand matching graph $G(\mcD) = (V', E')$ as follows:
\begin{itemize}
    \item  \textbf{Vertices:} $H$ has vertices $V' = \bigsqcup_v \copies(v)$ where $\copies(v)$ is $A(v)$ unique ``copies'' of $v$.
    \item \textbf{Edges:} For each demand $D_i$, let $E_i$ be any matching where the number of edges between $\copies(u)$ and $\copies(v)$ for each $u,v \in V$ is $D_i(u,v)$. Then $E' = \bigcup_i E_i$.
\end{itemize}
\end{definition}

\noindent We observe that the demand-matching graph is $s$-parallel greedy.
\begin{restatable}{lemma}{demandArb} \label{lem:arbBound} Let $C_1, C_2, \ldots, C_k$ be a sequence of $(h, s)$-length cuts with respect to node weighting $A$ with witness demands $\mcD = (D_1, D_2, \ldots, D_k)$. Then the demand-matching graph $G(\mcD)$ is an $|A|$-node $s$-parallel-greedy graph.
\end{restatable}
\begin{proof}
$G(\mcD)$ has $|A|$ nodes by construction. 

It remains to show that $G(\mcD)$  is an $s$-parallel-greedy graph as defined in \Cref{dfn:PGGraph}. By construction, $G(\mcD)$ is the union of a sequence of matchings $E_1, E_2, \ldots, E_k$ with corresponding demands $D_1, D_2, \ldots, D_k$ as described in \Cref{def:demandMatching}. For each $i$, we let $G_{i}$ be the graph resulting from the union of all matchings $E_k, E_{k-1}, \ldots, E_i$ (note that the matchings are in reverse order). Consider an edge $e = \{x_0',x_s'\}$ in $E_{i-1}$, the matching corresponding to demand $D_{i-1}$. By definition of $s$-parallel-greedy graphs, it suffices to show that there is no path from $u$ to $v$ in $G_i$ of length at most $s$. 

Assume for contradiction that such a path $P=(x_0',x_1',\ldots,x_{s-1}',x_s')$ existed in $G_{i}$ where each $x_l' \in \copies(x_l)$ for some $x_l \in V$.
By definition of the demand matching graph and the fact that applying cuts only increases distance, we have that every pair $(x_l, x_{l+1})$ is at distance at most $h$ in graph $G-\sum_{j\le i-1}C_{j}$. By triangle inequality, this implies that the distance between $x_0$ and $x_s$ in $G-\sum_{j\le i-1}C_{j}$ is less than $hs$. However, since $C_{i-1}$ is a cut that $hs$-separates all pairs in the support of $D_{i-1}$ which includes the pair $\{x_0, x_s\}$, we have that the distance between $x_0$ and $x_s$ in $G-\sum_{j\le i-1}C_{j}$ is strictly larger than $hs$, a contradiction.
\end{proof}

\subsection{Matching-Dispersed Demand}
In the previous section, we formalized the graph induced by the witness demands $(D_1, D_2, \ldots)$ of a sequence of length-constrained cuts. We now discuss how to use a forest decomposition (see \Cref{sec:conventions}) of this graph to construct a new demand by ``dispersing'' the demands $(D_1, D_2, \ldots)$ so that the result is of similar size to the original demands.

\begin{figure}
    \centering
    \begin{subfigure}[b]{0.48\textwidth}
        \centering
        \includegraphics[width=\textwidth,trim=0mm 0mm 0mm 0mm, clip]{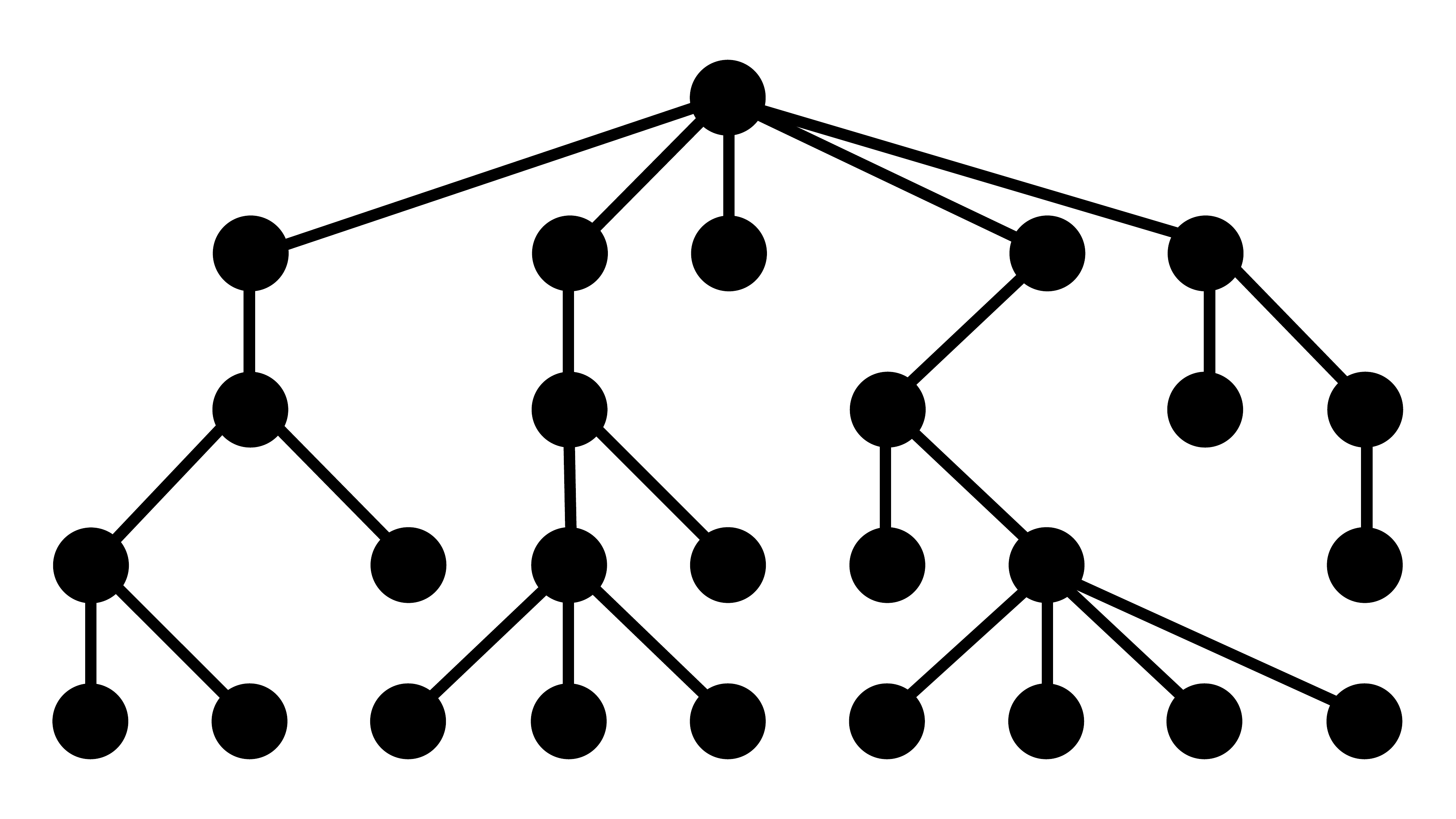}
        \caption{Input tree $T$.}\label{sfig:dispTree1}
    \end{subfigure}    \hfill
    \begin{subfigure}[b]{0.48\textwidth}
        \centering
        \includegraphics[width=\textwidth,trim=0mm 0mm 0mm 0mm, clip]{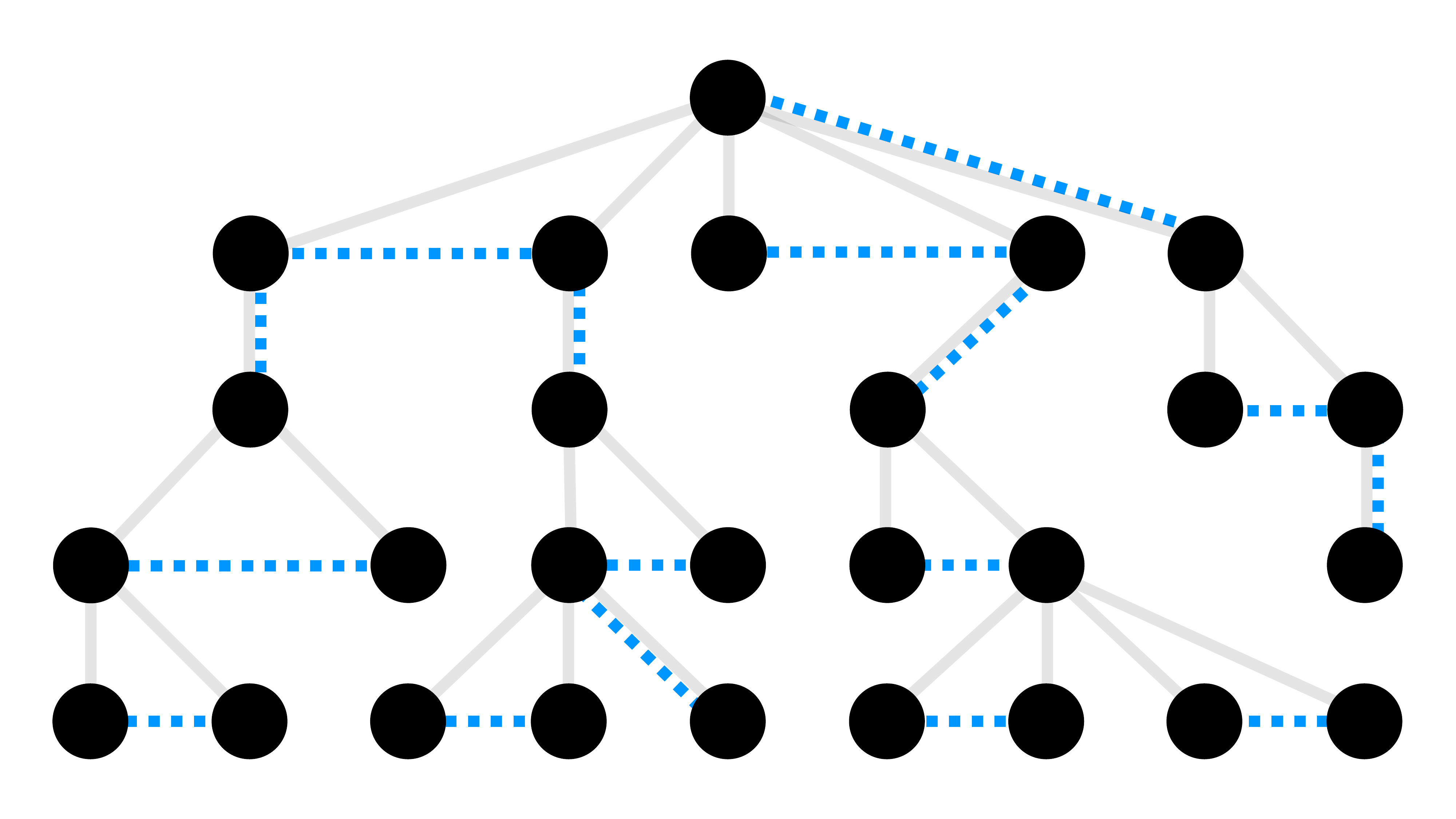}
        \caption{$\disperse_T$.}\label{sfig:dispTree2}
    \end{subfigure}
    \caption{How we disperse demand given a tree $T$ (\ref{sfig:dispTree1}). \ref{sfig:dispTree2} gives the support of $\disperse_T$ dashed in blue; notice that each vertex has degree at most $2$.}\label{fig:disperseTree}
\end{figure}

The following notion of a tree matching demand formalizes how we disperse the demand in each tree of the forest decomposition of the demand matching graph. Informally, given a tree this demand simply matches siblings in the tree to one another. If there are an odd number of siblings, the leftover child is matched to its parent. See \Cref{fig:disperseTree} for an illustration.

\begin{definition}[Tree Matching Demand]\label{def:treeMatchDemand}
Given tree $T = (V,E)$, we define the tree-matching demand on $T$ as follows. Root $T$ arbitrarily. For each vertex $v$ with children $C_v$ do the following. If $|C_v|$ is odd let $U_v = C_v \cup \{v\}$, otherwise let $U_v = C_v$. Let $M_v$ be an arbitrary perfect matching on $U_v$ and define the demand associated with $v$ as 
\begin{align*}
    D_v(u_1, u_2) := 
    \begin{cases}
        1 & \text{if $\{u_1,u_2\} \in M_v$}\\
        0 & \text{otherwise}.
    \end{cases}
\end{align*}
where each edge in $M_v$ has an arbitrary canonical $u_1$ and $u_2$. The tree matching demand for $T$ is
\begin{align*}
    \disperse_{T} := \sum_{v \text{ internal in }T} D_v
\end{align*}
\end{definition}

We observe that a tree matching demand has size equal to the input size (up to constants).
\begin{lemma} \label{lem:treeMatchSize}
Let $T$ be a tree with $n$ vertices. Then $|\disperse_T| \geq \frac{n-1}{2}$.
\end{lemma}
\begin{proof}
    For each $v$ that is internal in $T$ let the vertices $U_v$ and the perfect matching $M_v$ on $U_v$ be as defined in \Cref{def:treeMatchDemand}. Then, observe that $\sum_v |U_v| \geq n-1$ since every vertex except for the root appear in at least one $U_v$. On the other hand, for each $v$ since $M_v$ is a perfect matching on $U_v$ we have $|M_v| = \frac{1}{2} |U_v|$ and since $|\disperse_T| = \sum_{v \text{ internal in } T} |M_v|$, it follows that $|\disperse_T| \geq \frac{n-1}{2}$ as required.
\end{proof}

Having formalized how we disperse a demand on a single tree with the tree matching demand, we now formalize how we disperse an arbitrary demand by taking a forest cover, applying the matching-dispersed demand to each tree and then scaling down by the arboricity.

\begin{definition}[Matching-Dispersed Demand]\label{dfn:matchingDemand}
Given graph $G$, node-weighting $A$ and demands $\mcD = (D_1, D_2, \ldots)$, let $G(\mcD)$ be the demand matching graph (\Cref{def:demandMatching}), let $T_1, T_2, \ldots$ be the trees of a minimum size forest cover with $\alpha$ forests of $G(\mcD)$ and let $\disperse_{T_1}, \disperse_{T_2}, \ldots$ be the corresponding tree matching demands (\Cref{def:treeMatchDemand}). Then, the matching-dispersed demand on nodes $u,v \in V$ is
\begin{align*}
    \disperse_{\mcD, A}(u,v) := \frac{1}{2\alpha} \cdot \sum_i \sum_{u' \in \copies(u)}  \sum_{v' \in \copies(v)}\disperse_{T_i}(u', v').
\end{align*}
\end{definition}

We begin with a simple helper lemma that observes that the matching-dispersed demand has size essentially equal to the input demands (up to the arboricity).
\begin{lemma}\label{lem:matchingDemandSize}
    Given graph $G$, node-weighting $A$ and and demands $\mcD = (D_1, D_2, \ldots)$ where $G(\mcD)$ has arboricity $\alpha$, we have that the matching-dispersed demand $\disperse_{\mcD, A}$ satisfies $|\disperse_{\mcD, A}| \geq \frac{1}{4 \alpha} \sum_i |D_i|$.
\end{lemma}
\begin{proof}
    Observe that the number of edges in $G(\mcD)$ is exactly $\sum_i |D_i|$ and so summing over each tree $T_j$ in our forest cover and applying \Cref{lem:treeMatchSize} gives
\begin{align*}
    \sum_j |\disperse_{T_j}| \geq \frac{1}{2} \cdot \sum_i |D_i|.
\end{align*}
Combining this with the definition of $\disperse_{\mcD, A}$ (\Cref{dfn:matchingDemand}) gives the claim.
\end{proof}

We now argue the key properties of the matching-dispersed demand which will allow us to argue that it can be used as a large demand separated by the union of our cuts. 
\begin{lemma}[Properties of Matching-Dispersed Demand]\label{lem:sparseOfMatching}
    Given graph $G = (V, E)$ and node-weighting $A$, let $C_1, C_2, \ldots$ and $\mcD= (D_1, D_2, \ldots)$ be a sequence of length-constrained cuts and $A$-respecting $h$-length demands respectively where $C_i$ $hs$-separates $D_i$ in $G - \left(\sum_{l < i} C_i \right)_{hs}$. Then, if $G(\mcD)$ has arboricity $\alpha$, the matching dispersed demand $\disperse_{\mcD, A}$ is
    \begin{enumerate}
        \item \textbf{Length-Constrained Respecting:} a $2h$-length $A$-respecting demand;
        \item \textbf{Separated:} $h \cdot (s-1)$-separated by $\left(1+\frac{1}{s-1}\right)\sum_i C_i$;
        \item \textbf{Large:} of size $|\disperse_{\mcD,A}| \geq \frac{1}{4 \alpha} \sum_i A_{(h,s)}(C_i)$.
    \end{enumerate}
\end{lemma}
\begin{proof}
To see that $\disperse_{\mcD, A}$ is $2h$-length observe that vertices $u$ and $v$ have $\disperse_{\mcD, A}(u,v) > 0$ only if there is a path consisting of at most two edges between a node in $\copies(u)$ and a node in $\copies(v)$ in the demand matching graph $G(\mcD)$ (\Cref{def:demandMatching}). Furthermore, $u' \in \copies(u)$ and $v' \in \copies(v)$ have an edge in $G(\mcD)$ only if there is some $i$ such that $D_i(u, v) > 0$ and since each $D_i$ is $h$-length, it follows that in such a case we know $d_G(u,v) \leq h$. Thus, it follows by the triangle inequality that $\disperse_{\mcD, A}$ is $2h$-length.
    
To see that $\disperse_{\mcD, A}$ is $A$-respecting we observe that each vertex in $G(D)$ is incident to at most $2\alpha$ matchings across all of the tree matching demands we use to construct $\disperse_{\mcD, A}$ (at most $2$ matchings per forest in our forest cover). Thus, for any $u \in V$ since $|\copies(u)| = A(u)$ we know
\begin{align*}
    \sum_{u' \in \copies(u)} \sum_j  \sum_{v} \sum_{v' \in \copies(v)}\disperse_{T_j}(u', v') \leq \sum_{u' \in \copies(u)} 2\alpha \leq 2\alpha \cdot A(u).
\end{align*}

It follows that for any $u \in V$  we have
\begin{align*}
\sum_v \disperse_{\mcD, A}(u, v) = \sum_v \frac{1}{2\alpha}\sum_j \sum_{u' \in \copies(u)}  \sum_{v' \in \copies(v)}\disperse_{T_j}(u', v') \leq A(u)
\end{align*}
A symmetric argument shows that $\sum_v \disperse_{\mcD, A}(v, u) \leq A(u)$ and so we have that $\disperse_{\mcD, A}$ is $A$-respecting.

We next argue that $\disperse_{\mcD, A}$ is  $h(s-1)$-separated by $\left(1+\frac{1}{s-1}\right)\sum_i C_i$. Consider an arbitrary pair of vertices $u$ and $v$ such that $\disperse_{\mcD, A}(u,v) > 0$; it suffices to argue that $u$ and $v$ are $h(s-1)$-separated by $\left(1+\frac{1}{s-1}\right)\sum_i C_i$. As noted above, $\disperse_{\mcD, A}(u,v) > 0$ only if there is a path $(u', w', v')$ in $G(\mcD)$ where $u' \in \copies(u)$, $v' \in \copies(v)$ and for some $w \in V$ we have $w' \in \copies(w)$. But, $\{u', w'\}$ and $\{w', v'\}$ are edges in $G(\mcD)$ only if there is some $i$ and $j$ such that $D_i(u,w) > 0$ and $D_j(w,v) > 0$. 

By definition of $G(\mcD)$ (\Cref{def:demandMatching}), each $D_i$ corresponds to a different matching in $G(\mcD)$ and so since $\{u', w'\}$ and $\{w', v'\}$ share the vertex $w'$, we may assume $i \neq j$ and without loss of generality that $i < j$. Let $C_{\leq i} := \sum_{l \leq i} C_l$ and let $G_{\leq i} = G - (C_{\leq i})_{hs}$ be $G$ with $C_{\leq i}$ applied.

Since $D_i$ is $hs$-separated by $C_{\leq i}$ and $D_i(u,w) > 0$, we know that
\begin{align}
    d_{G_{\leq i}}(u, w) > hs.\label{eq:a}
\end{align}
On the other hand, since $D_j$ is an $h$-length demand, $j > i$ and $D_j(w,v) > 0$, we know that the distance between $w$ and $v$ in $G_{\leq i}$ is
\begin{align}
    d_{G_{\leq i}}(w, v) \leq h. \label{eq:b}
\end{align}
We claim it follows that
\begin{align}\label{eq:gas}
    d_{G_{\leq i}}(u, v) > h \cdot (s-1).
\end{align}
Suppose otherwise that $d_{G_{\leq i}}(u, v) \leq h \cdot (s-1)$. Then combining this with \Cref{eq:b} and the triangle inequality we would have that $d_{G_{\leq i}}(u, w) \leq hs $. However, this would contradict  \Cref{eq:a} and so we conclude that $d_{G_{\leq i}}(u, v) > h \cdot (s-1)$.

Continuing, by definition of $d_{G_{\leq i}}$ and \Cref{eq:gas}, we know that the distance between $u$ and $v$ according to to the length function $l +hs \cdot \sum_{l \leq i} C_l$ is at least $h \cdot (s-1)$ (where $l$ is the original length function of $G$). In order to show that $u$ and $v$ are $h(s-1)$-separated by $\left(1+\frac{1}{s-1}\right)\sum_i C_i$, it suffices to show that $u$ and $v$ are $h(s-1)$-separated by $\left(1+\frac{1}{s-1}\right)\sum_{l \leq i} C_l$ and, in particular, that $u$ and $v$ are at distance at least $h \cdot (s-1)$ according to the length function $l + h  (s-1) \cdot \left(1+\frac{1}{s-1}\right)\sum_{l \leq i} C_l$. However, this length function is equal to $l + hs \cdot \sum_{l \leq i} C_l$ and we already know that the distance between $u$ and $v$ according to this length function is at least $h \cdot (s-1)$. This shows that $u$ and $v$ are $h(s-1)$-separated by $\left(1+\frac{1}{s-1}\right)\sum_i C_i$.


Lastly, we argue that $|\disperse_{\mcD,A}| \geq \frac{1}{4 \alpha} \sum_i A_{(h,s)}(C_i)$. By \Cref{lem:matchingDemandSize} we know that 
\begin{align*}
    |\disperse_{\mcD, A}| \geq \frac{1}{4\alpha} \sum_i |D_i|.
\end{align*}
Combining this with the fact that $A_{(h,s)}(C_i) = |D_i|$ gives us
\begin{align*}
	|\disperse_{\mcD,A}| \geq \frac{1}{4 \alpha} \sum_i A_{(h,s)}(C_i)
\end{align*}
 as required.
 \end{proof}
 
 \subsection{Completing Union of Length-Constrained Cuts via Parallel Greedy Arboricity}
 Lastly, we complete the proof of \Cref{lem:UCviaPGArb}.
\UCviaPGArb*
\begin{proof}
    We begin by unpacking the relevant definitions.
	Recall that the $(h',s')$-length sparsity of length-constrained cut $\left(1+\frac{1}{s-1}\right)\sum_i C_i$ with respect to node-weighting $A$ is (by \Cref{def:sparsity}) defined as
	\begin{align*}
        \spa_{(h',s')}\left(\left(1+\frac{1}{s-1}\right)\sum_i C_i,A\right) =  \frac{\left(1+\frac{1}{s-1} \right)|\sum_i C_i|}{A_{(h',s')}\left(\left(1+\frac{1}{s-1} \right)\sum_i C_i\right)}.
	\end{align*}
	Thus, to demonstrate that $\left(1+\frac{1}{s-1}\right)\sum_i C_i$ has $(h',s')$-length sparsity at most $\phi'$, we must show that $ \frac{\left(1+\frac{1}{s-1} \right)|\sum_i C_i|}{A_{(h',s')}\left(\left(1+\frac{1}{s-1} \right)\sum_i C_i\right)}$ is at most $\phi'$. In other words, we must demonstrate that the $(h',s')$-length demand-size (\Cref{def:demandSize}) of $\left(1+\frac{1}{s-1}\right)\sum_i C_i$ is at least $\left(1+\frac{1}{s-1}\right)|\sum_i C_i|/\phi'$. Doing so requires that we demonstrate the existence of an $A$-respecting $h'$-length demand $D$ which is $h's'$-separated by $\left(1+\frac{1}{s-1}\right)\sum_i C_i$ and of size at least $\left(1+\frac{1}{s-1}\right)|\sum_i C_i|/\phi'$. Since we have assumed $s \geq 2$, it suffices for the size of such a demand to be at least $2|\sum_i C_i|/\phi'$.  For the remainder of this proof, we demonstrate that the matching-dispersed demand is exactly such a demand.
	
    Let demands $\mcD = (D_1, D_2, \ldots)$ be the demands which witness the $(h,s)$-demand-size of our cuts after the preceding cuts have been applied. In particular, for each $i$ let $D_i$ be any demand that is $hs$-separated by $C_i$ in $G - \left( \sum_{j < i} C_i \right)_{hs}$ such that $|D_i|= A_{(h,s)}(C_i)$ where $A_{(h,s)}(C_i)$ is the demand-size of $C_i$ in $G - \left( \sum_{j < i} C_i \right)_{hs}$. Such a demand exists by definition of demand-size (\Cref{def:demandSize}). Likewise, let $\disperse_{\mcD, A}$ be the matching-dispersed demand as defined in \Cref{dfn:matchingDemand}. By \Cref{lem:arbBound}, we know that $G(\mcD)$ is an $|A|$-node $s$-parallel-greedy graph. Since we have assumed that all $n$-node $s$-parallel greedy graphs have arboricity at most $\alpha(n,s)$, it follows that $G(\mcD)$ has arboricity at most $\alpha(|A|,s)$. Combining this with \Cref{lem:sparseOfMatching}, we get that $\disperse_{\mcD, A}$ is a $2h$-length $A$-respecting demand that is $h \cdot (s-1)$ separated by $\left(1+\frac{1}{s-1}\right)\sum_{i} C_i$ and of size at least $|\disperse_{\mcD,A}| \geq \frac{1}{4  \alpha(|A|,s)} \sum_i A_{(h,s)}(C_i)$. Equivalently, since $h' = 2h$ and $s' = \frac{s-1}{2}$, we have that  $\disperse_{\mcD, A}$ is an $h'$-length $A$-respecting demand that is $h's'$-separated by $\left(1+\frac{1}{s-1}\right)\sum_{i} C_i$. Likewise, by our choice of $\phi' = 8\alpha(|A|,s) \cdot \frac{\sum_i |C_i|}{\sum_i A_{(h,s)}(C_i)}$ we get that the size of $\disperse_{\mcD,A}$ is 
	\begin{align*}
		|\disperse_{\mcD,A}| \geq \frac{1}{4  \alpha(|A|,s)} \sum_i A_{(h,s)}(C_i) = \frac{2\sum_i |C_i|}{\phi'}.
	\end{align*}
	as required.
\end{proof}

\bibliographystyle{alpha}
\bibliography{main}

\end{document}